\documentclass[11pt]{article}
\usepackage[margin = 1.5in]{geometry}
\usepackage{mathpazo}\usepackage{amssymb}
\usepackage{amsmath}
\usepackage{amsthm}

\usepackage{tikz}
\usetikzlibrary{decorations.pathreplacing}

\usepackage{mathrsfs}
\newcommand{\scr}{\mathscr}
\newcommand{\sfrac}[2]{{\textstyle\frac{#1}{#2}}}
\newcommand{\diag}{\operatorname{diag}}
\newcommand{\degree}{d}

\newcommand{\poly}{\mathit{poly}}

\newtheorem{lem}{Lemma}
\newtheorem{thm}{Theorem}

\begin{document}

\title{The Parity of Directed Hamiltonian Cycles}
\date{}
\author{Andreas Bj\"orklund and Thore Husfeldt\footnote{Lund
    University, Sweden and IT University of Copenhagen, Denmark. This
    work is partially supported by the Swedish Research Council, grant VR 2012-4730.}}
\maketitle
\begin{abstract}
  We present a deterministic algorithm that given any directed graph
  on $n$ vertices computes the parity of its number of Hamiltonian
  cycles in $O(1.619^n)$ time and polynomial space.
  For bipartite graphs, we give a $1.5^n\poly(n)$
  expected time algorithm.

  Our algorithms are based on a new combinatorial formula for the
  number of Hamiltonian cycles modulo a positive integer.
\end{abstract}

\section{Introduction}

It is known since the 1960s that Hamiltonian cycles in an $n$-vertex
graph can be detected and counted in $O(2^nn^2)$ time~\cite{B62,HK62}.
In an influential survey, Woeginger~\cite{W03} asked if this could be
significantly improved.
Recently, two results using different techniques have met this
challenge:
\begin{enumerate}
\item a randomized $O(1.657^n)$ time algorithm for the decision
  problem in undirected graphs~\cite{Bjorklund10},
\item a randomized $O(1.888^n)$ time algorithm for the decision problem
  in directed bipartite graphs~\cite{Cygan12}.
\end{enumerate}
For the general directed graph case, no such algorithm has yet appeared.

Intriguingly, the foundation for the construction from \cite{Cygan12}
is an algorithm for the related $\oplus \mathrm P$-complete problem of
computing the parity of the number of Hamiltonian cycles (in a
bipartite graph).
Our contribution in the present paper is to compute this value (for
any directed graph) in time within a polynomial factor of $O(\phi^n)$,
where $\phi=\frac{1}{2}(1+\surd 5)<1.619$ is the golden ratio.

While this does not seem to shed any direct light on the decision
problem in directed graphs, it does however ask a related interesting
question: Could it really be easier to solve a $\oplus \mathrm
P$-complete problem than its decision counterpart?
There is still no known $(2-\Omega(1))^n$ time algorithm for the decision problem
and even the recent decision algorithm for 
undirected graphs \cite{Bjorklund10} is slower than our present
algorithm.

In contrast, current evidence for the well-studied CNF Satisfiability problem
points in the opposite direction: It is known that a fast algorithm
for computing the parity of the number of satisfiable assignments to a CNF formula on $n$ Boolean variables would disprove
the Strong Exponential Time Hypothesis~\cite{Cyganetal}, i.e. a $(2-\Omega(1))^n$ time algorithm
computing the parity of the satisfying assignments implies a $(2-\Omega(1))^n$ time algorithm for deciding if a formula has
a satisfying assignment at all.

\subsection{Results}

Let $\scr H$ denote the set of (directed) Hamiltonian
cycles of a directed input graph $G=(V,E)$ on $n=|V|$ vertices.

Our main result is that the parity $\oplus\scr H$ of the number of
directed Hamiltonian cycles can be computed much faster than in $2^n$
time.

\begin{thm}\label{thm: main}
  We can compute  $\oplus\scr H$ for a general $n$-vertex digraph in time within a polynomial
  factor of the $n$th Fibonacci number $F_n\in O(1.619^n)$ and
  polynomial space.
\end{thm}

For the restricted family of bipartite graphs, we can show a stronger
bound.
This result is also somewhat easier to prove.

\begin{thm}\label{thm: bip}
  We can compute $\oplus\scr H$ for a bipartite $n$-vertex digraph in expected time within 
  a polynomial factor of $O(1.5^n)$ and polynomial space.
\end{thm}

Both results rely on a new characterization of the number of Hamiltonian
cycles in terms of a local property.

\begin{thm}
\label{thm: comb}
For a vertex subset $X$, let $\degree_v(X)$ denote the number of
directed edges from
vertex~$v$ to a vertex in $X$.
For integer $K\geq 2$,
  \begin{equation*}
  |\scr H|  = \frac{1}{K} \sum_{Z,Y_1,\ldots,Y_K} 
(-1)^{|V\setminus Z|}
  \biggl(\prod_{z\in Z}  \degree_z(V\setminus Z) \biggr)
  \prod_{k=1}^K
  \biggl(\prod_{y\in Y_k}  \degree_y(Y_k) \biggr)
   \pmod K\,,\end{equation*}
where the sum is over all $(K+1)$-partitions of $V$.
\end{thm}%

\noindent Our algorithms use this result for $K=2$, but we state (and prove) it
for general $K$.

\subsection{Related results}

There are nontrivial examples of hard counting problems where
computing the parity is easy.
In fact an important example is given by a closely related problem:
computing the number $\#\scr C$ of disjoint cycle covers in a directed
graph.
This is a hard counting problem, equivalent to computing the permanent
of the adjacency matrix.
On the other hand, the parity $\oplus\scr C$ equals the parity of the
determinant and is therefore computable in polynomial time by Gaussian
elimination.
However, the problem studied in the present paper, $\oplus\scr H$,
is in fact complete for the complexity class $\oplus\mathrm P$, even for
very restricted classes of graphs~\cite{Valiant2005}.

Counting the number of Hamiltonian cycles in an $n$-vertex directed
graph can be done in $o(2^n)$ time \cite{Bjorklund12}.
However, the improvement over the classic algorithms~\cite{B62,HK62}
is relatively small (within a factor of
$\operatorname{exp}(O(\sqrt{n/\log n}))$).
The parity problem appears to be easier.
First, for undirected graphs it is known that when all vertex degrees
are odd, there is an even number of Hamiltonian cycles passing through
each (directed) edge~\cite{T78}.
Second, a recent paper of Cygan, Kratsch, and Nederlof~\cite{Cygan12}
includes a deterministic $O(1.888^n)$ time, exponential space
algorithm computing the parity of the number of Hamiltonian cycles for
undirected and directed bipartite graphs.
Our present algorithm is faster, uses only polynomial space, and works
for unrestricted directed graphs as well.
On the other hand, the algorithm from \cite{Cygan12} works for
weighted graphs.
This is a crucial property because it allows the use of the Isolation
lemma to construct a (randomized) decision algorithm.
Our constructions do not seem to allow this extension.


\subsection{Open questions}

Our paper accentuates the fact that the exponential time complexity of
counting or deciding the Hamiltonian cycles in a directed graph is not
very well understood.
While there are several examples of $\oplus \mathrm P$-complete
problems whose decision analogue is computationally easier
(e.g., $\oplus$2-Satisfiability), examples of the converse
are not known to the authors.

Under the Exponential Time Hypothesis, the decision problem does not
allow $\exp(o(n))$-time algorithms,
but there are no arguments for or against an $O(1.999^n)$-time algorithm
for the decision problem.

For comparison, it is now known that the parity of the number of set
covers of a given set system on $n$ elements cannot be solved in time
$(2-\Omega(1))^n$ under the Strong Exponential Time Hypothesis.
\cite{Cyganetal}.

\subsection{Overview and techniques}

Our algorithm is based on evaluating theorem~\ref{thm: comb} with $K=2$, which
is a summation of all 3-partitions of $G$.

There are two parts to our algorithm.

First, the vast majority of the terms in the summation can be made to
vanish modulo 2.
Our main algorithmic insight is the fact that any set of self-loops, and hence in particular a random set, can be
added to the vertices of a Hamiltonicity instance without changing the
result.
This idea of randomly modifying the input instance in order to produce
a large number of zero-valued terms is inspired by an algorithm for
computing the permanent by Bax and Franklin \cite{BaxFranklin}.
After this modification, we can avoid explicitly listing the
vanishing terms.
The running time of the algorithm is based on a
careful orchestration of a total search among linear equation systems,
which happens in sections~\ref{sec: bip} and ~\ref{sec: genalg} for
the bipartite and the general input case, respectively.

Second, the contribution of the nonvanishing terms can be computed in
polynomial time via linear algebra, see section~\ref{sec: alg}.

Finally, in section~\ref{sec: comb} we give the proof of
theorem~\ref{thm: comb}, the combinatorial core of our contribution.
This argument at least follows a well-trodden path, if only in the
beginning; the cut-and-count method \cite{Cygan11} provides a connection between
Hamiltonicity and cycle covers, and the principle of
inclusion--exclusion is used to move from sums over permutations to
sums of functions.

\section{Algorithm}
\label{sec: alg}

We begin by rewriting the expression from theorem~\ref{thm: comb}.
First, for $K=2$, we get
\begin{equation}
\label{eq: K=2}
\oplus\scr H = \sfrac{1}{2} \sum_{X,Y,Z}
\biggl(\prod_{x\in X} \degree_x(X) \biggr)
\biggl(\prod_{y\in Y}\degree_y(Y) \biggr) 
\biggl(\prod_{z\in Z} \degree_z(\overline Z) \biggr)
\pmod 2\,,
\end{equation}
where we have introduced $X=Y_1$, $Y=Y_2$ and let $\overline Z$ denote
the vertex complement $\overline Z= V\setminus Z$.
Note in particular that we removed the factor $(-1)^{|\overline Z|}$,
because $-1=+1\pmod 2$, and even though the mod $2$ operation is applied after the division by $2$, every contributing
$Z$ will be counted an even number of times as seen by changing the roles of $X$ and $Y$.

It will be convenient to remove the factor $\sfrac{1}{2}$ altogether.
To this end, consider a binary relation $\prec$ on the subsets of $V$
such that for disjoint $X,Y\subseteq V$, not both empty, either
$X\prec Y$ or $Y\prec X$.
Then \eqref{eq: K=2} is equivalent to
\begin{equation}\label{eq: prec introduced}
\oplus\scr H =  \sum_{\substack{X,Y,Z\\X\prec Y}}
\biggl(\prod_{x\in X} \degree_x(X) \biggr)
\biggl(\prod_{y\in Y}\degree_y(Y) \biggr) 
\biggl(\prod_{z\in Z} \degree_z(\overline Z) \biggr)
\pmod 2\,.
\end{equation}
To see this, first observe that~\eqref{eq: K=2} is symmetric with
respect to $X$ and $Y$.
Furthermore, in a contributing term, the sets $X$ and $Y$ cannot both
be empty, because for $Z=V$ the term $\prod_{z\in Z}
\degree_z(\overline{Z})$ vanishes.
Thus for every partition $(X,Y,Z)$ with $X\prec Y$ that contributes to
\eqref{eq: prec introduced}, there is a twin contribution to
\eqref{eq: K=2}: the distinct partition $(Y,X,Z)$, which has $Y\prec
X$ and in particular $Y\neq X$.

Finally, we separate the expression into the contributions of $X$ and
$Y\cup Z$, respectively:
\begin{equation}\label{eq: factored}
\oplus\scr H =\sum_X f(X)   \prod_{x\in X} \degree_x(X) \pmod 2\,,
\end{equation}
where
\begin{equation}\label{eq: f def}
f(X)=
  \sum_{\substack{Y,Z\\X \prec Y}}
  \biggl(\prod_{y\in Y}  \degree_y(Y) \biggr)
  \prod_{z\in Z}  \degree_z(\overline Z)\,.
\end{equation}
It is understood that the sum in \eqref{eq: factored} is over all
$X\subseteq V$ and the sum in \eqref{eq: f def} is over all partitions
$Y,Z$ of $V\setminus X$.
Still, the whole expression is over almost all $3$--partitions $X,Y,Z$
of the vertices and thus has $\Omega(3^n)$ terms, so it does not in
itself serve as a fast algorithm.

\bigskip
There are two parts to our algorithm.
First, for many $X$, already the factor $\prod_{x\in X} \degree_x(X)$ in
\eqref{eq: factored} will be zero modulo $2$.
Hence, these subsets contribute nothing and we will avoid explicitly
listing them.
Second, for any $X$ for which it is not zero, the value of $f(X)$
modulo $2$ can be computed in polynomial time in $|V|$ via linear
algebra.

On the top level our algorithm is simply an evaluation of \eqref{eq:
  factored} and \eqref{eq: f def}.

\bigskip
\noindent {\bf Algorithm P} (\emph{Parity}.) Given a 
directed graph $G=(V,E)$,  computes $\oplus \scr H$.
\begin{description}\setlength\itemsep{0pt}
\item[P1][Initialise.] Set $s=0$.
\item[P2] [Locate.] List every $X\subseteq V$ such that $\prod_{x\in X} \degree_x(X)$ is odd.
\item[P3] [Contribute.] Compute $f(X) \pmod 2$ for every such $X$ and add it to $s$.
\item[P4][Report.] Return $s \pmod 2$.
\end{description}
Note that steps P2 and P3 can be interleaved in order to avoid
storing all $X$.

We proceed by explaining how to execute steps P2 and P3 efficiently.
 
\subsection{Step P2: Locating Subsets that Contribute}

We say that a subset $X\subseteq V$ is \emph{contributing} if
$\prod_{x\in X} \degree_x(X)$ is odd.
Our aim is to generate the contributing subsets.
We identify $V$ with $\{1,\ldots, n\}$ and introduce the indicator
variables $x_1,\ldots, x_n\in\{0,1\}$ for $X\subseteq V$ with $x_i=1$
if and only if $i\in X$.
Then the constraint that $\degree_i(X)$ is odd can be expressed in terms of
the adjacency matrix $A$ of $G$ as
\[ \sum_{j=1}^n a_{ij}x_j =1 \pmod 2\qquad \text{for all } i\in X\,,\]
equivalently,
\[ x_i\biggl(\sum_{j=1}^n a_{ij}x_j\biggr) =x_i \pmod 2\qquad
\text{for all } i=1,\ldots, n\,.\]

We will view these constraints as a system of $n$ quadratic equations
over GF(2) in the variables $x_1,\ldots,x_n$:
\begin{alignat*}{7}
  x_1(a_{11} x_1 &+a_{12}x_2 &+\cdots +a_{1n}x_n) &= x_1\\
  x_2(a_{21} x_1 &+a_{22}x_2 &+\cdots +a_{2n}x_n) &= x_2\\
&&&\vdots&\\
  x_n(a_{n1} x_1 &+a_{n2}x_2 &+\cdots +a_{nn}x_n) &= x_n\,.
\end{alignat*}
This system can be succinctly expressed as a matrix equation over GF(2):
\begin{equation}\label{eq: quadratic eqs}
\mathbf x\circ A\mathbf x=\mathbf x
\end{equation}
where $\mathbf x=(x_1,\ldots, x_n)$ and the operator $\circ$
denotes the coordinate-wise (or Schur, or Hadamard) product.

\subsection{Random self-loops}

\emph{A priori}, listing the solutions to \eqref{eq: quadratic eqs}
requires $2^n$ steps.
In fact, already the size of the solution set is easily seen to be of
order $2^n$ for some graphs.
For instance, every vector $\mathbf x$ with an odd number of $1$s
solves \eqref{eq: quadratic eqs} if $A$ is the all-1s matrix corresponding to
a directed clique with self-loops.

This motivates the main algorithmic insight in our construction.
Because no Hamiltonian cycle can use a self-loop, we can add
self-loops to the vertices in $G$ (or remove them) without changing $\scr
H$.
Algebraically, manipulation of self-loops corresponds to flipping the
diagonal entries of $A$.

We show that adding these loops at random reduces the expected
solution size for all graphs.

\begin{lem}\label{lem: expected size}
  Let $A$ be an $n\times n$ matrix with $0,1$-entries.
  Choose the diagonal entries $a_{11},\ldots, a_{nn}\in \{0,1\}$
  uniformly and independently at random.
  Then the expected number of solutions to \eqref{eq: quadratic eqs} is $1.5^n$.
\end{lem}

\begin{proof}
  Consider the $i$th entry of $\mathbf x$.
  If $x_i=0$ then the $i$th equation in  \eqref{eq:
    quadratic eqs} is $0=0$ and trivially satisfied.
  If $x_i=1$ then the $i$th equation in \eqref{eq:
    quadratic eqs} is satisfied if
  \[
  1 = \sum_{j=1}^n a_{ij}x_j  = a_{ii} \cdot 1 + \sum_{\substack{j=1\\j
\neq i}}^n a_{ij}x_j  \,,
  \]
  which happens with probability $\frac{1}{2}$ because $a_{ii}$ is
  uniformly distributed.
  By independence of the choices of $a_{ii}$, a vector $\mathbf x$
  with $k$ 1s
  satisfies all the equations with probability $2^{-k}$.
  By linearity of expectation, the expected  number of solutions is 
  \[
  \sum_{k=0}^n\binom{n}{k}2^{-k} = (1+\sfrac{1}{2})^n\,,
  \]
  using the binomial theorem.
\end{proof}

In section \ref{sec: bip} we show that these $1.5^n$ solutions can be
efficiently listed if $G$ is bipartite.
In section \ref{sec: genalg} we show that the solutions can be listed
in time $O(1.619^n)$ for general graphs.

\subsection{Gaussian elimination}

We assume that the following algorithm for solving nonhomogeneous
systems of linear equations is well known.
We repeat it here only to recall that the solution set can be
described in terms of the basis of a translated vector space (the
system's \emph{null space}).
In particular, if a linear equation system has a solution at all, the
number of solutions equals $2^d$, where $d$ is the dimension of the
null space.

\bigskip\noindent {\bf Algorithm G} (\emph{Gaussian elimination}.)
Given a nonhomogeneous system $Ax=c$ of linear equations.
If the system has no solution, outputs ``no''.
Otherwise, outputs a vector $v$ and $d$ linearly independent vectors
$b_1,\ldots,b_d$ such that the set of vectors $x$ with $Ax=c$ equals  $\{v +
\operatorname{span}(b_1,\ldots, b_d)\}$.

\bigskip
Gaussian elimination runs in polynomial time and works over
finite fields.

\subsection{Step P3: Computing the Contributions}

We return to the value $f(X)$ in \eqref{eq: f def} and show how to
compute it modulo 2 in polynomial time. The salient feature of 
\eqref{eq: f def} is that modulo 2, we can now
view $f(X)$ as the number of solutions to an equation system over GF(2).

\medskip
We define the binary relation $\prec$ on the subsets of $V$ as follows.
Identify $V$ with $\{1,\ldots, n\}$ and set $X\prec Y$ if $\min X <
\min Y$, with the usual convention $\min \emptyset =
\infty$.
(By the discussion following \eqref{eq: prec introduced}, the case $X=Y=\emptyset$ never arises.)

We will set up a system of linear equations in $y_1,\ldots,y_n$, so
that $f(X)$ equals the number of solutions (mod 2).
We use the obvious correspondence
\[ y_i =\begin{cases}
1\,,&\text{if $i\in Y$}\,;\\
0\,,&\text{if $i\notin Y$}\,.\\
\end{cases}
\]
First, since $Y$ belongs to the complement of $X$ we introduce the equations
\begin{equation}\label{eq: complement}
  y_i = 0\qquad \text{for } i\in X\,.
\end{equation}
Second, to ensure $X\prec Y$ we introduce the equations
\begin{equation}\label{eq: order}
 y_i = 0,\qquad \text{for }i=1,\ldots, \min X\,.
\end{equation}

Finally, a subset $Y$ with $X\prec Y$ contributes to $f(X)$ if all
product terms in \eqref{eq: f def} are
$1$, so we need 
\begin{align*}
  \degree_i(Y) &= 1 \pmod 2,\qquad \text{if } i\in Y\, \\
\degree_i(X\cup Y) &= 1 \pmod 2, \qquad \text{if } i\notin Y\,
\end{align*}
for each $i\notin X$. 
Since  $\degree_i(X\cup Y) = \degree_i(X) + \degree_i(Y)$ we can rewrite the second
constraint to 
\[
 \degree_i(Y) = 1 + \degree_i(X)  \pmod 2, \qquad \text{if } i\notin Y\,.
\]
In terms of the variables $y_1,\ldots, y_n$ and the adjacency matrix
$A$ of $G$ we have the equations
\[
\sum_{j=1}^n a_{ij} y_j = 
\begin{cases}
1\,,&\text{if $y_i= 1$}\,;\\
1+\degree_i(X)\,,&\text{if $y_i= 0$}\,;
\end{cases}
\]
for each $i\notin X$.
(Note that $\degree_i(X)$ is a fixed value depending only on $X$.)
The right hand side simplifies to $1+\degree_i(X) + \degree_i(X) y_i$, so the
resulting equations are
\begin{equation}\label{eq: neighbourhood}
\degree_i(X)y_i + \sum_{j=1}^n a_{ij} y_j = 1+\degree_i(X),\qquad \text{for }i\notin X\,.
\end{equation}

In summary, the solutions to the linear equations in \eqref{eq: complement},
\eqref{eq: order}, and \eqref{eq: neighbourhood} describe exactly the
subsets $Y\subseteq V\setminus X$ with $X\prec Y$ that contribute 1 to
$f(X)$.
The number of such solutions is $2^d$ with $d$ the dimension of the system's null space.
It is odd if and only if the system of equations
has a unique solution (and consequently $d=0$).
This is determined by Gaussian elimination in time polynomial in
$n$.

\section{Bipartite Input Graphs}
\label{sec: bip}

We first describe an algorithm for the step P2 for the case that the input graph is
bipartite.
We can safely assume that the input bipartite graph is balanced on two
equal sized vertex sets $V_1\cup V_2=V$.
(Otherwise, there can be no Hamiltonian cycles.)
 
Let 
\[
B= \begin{bmatrix}0 & B_{12} \\B_{21} & 0 \end{bmatrix} 
\]
denote the adjacency matrix of the input graph.
Choose random values $r_1,\ldots,r_n\allowbreak\in\{0,1\}$ uniformly and
independently and construct the diagonal matrix $R=\diag(r_1,\allowbreak\ldots,\allowbreak r_n)$. 
Set $A=B+R$.
The equation system \eqref{eq: quadratic eqs} becomes 
\[
\mathbf x \circ(B+R)\mathbf x= 
\mathbf x\circ
\begin{bmatrix}0 & B_{12} \\ B_{21} &  0\end{bmatrix} 
\mathbf x + \mathbf x\circ R\mathbf x = 
\mathbf x\,,
\]
equivalently,
\begin{align*}
\mathbf x_1\circ B_{12} \mathbf x_2 &= \mathbf x_1-\mathbf x_1\circ R_1 \mathbf x_1\\
\mathbf x_2\circ B_{21} \mathbf x_1 &= \mathbf x_2- \mathbf x_2\circ R_2 \mathbf x_2\,,
\end{align*}
where $\mathbf x =(\mathbf x_1, \mathbf x_2)$ and $R=
\diag(R_1,R_2)$.
In this formulation it is apparent that for fixed $\mathbf x_1$, the
equation system (in the variables $\mathbf x_2$) is linear; note in
particular that $\mathbf x_2\circ R_2 \mathbf x_2 =R_2
\mathbf x_2$ because $R_2$ is diagonal and multiplication in GF(2) is
idempotent.
Thus, the system can be solved by Gaussian elimination in polynomial
time.

\bigskip
\noindent {\bf Algorithm B} (\emph{Bipartite graphs}.) Given a bipartite
directed graph $G=(V,E)$ with $n=|V|$, lists the vertex subsets $X\subseteq V$
for which $\degree_x(X)$ is odd for all $x\in X$.
\begin{description}\setlength\itemsep{0pt}
\item[B1] [Add random self loops.] Choose $r_1,\ldots, r_n\in \{0,1\}$
  uniformly at random and add a self-loop at vertex $v_i$ if $r_i=1$.
\item[B2] [Initialise.] Set $x_1=\cdots =x_{n/2}=0$.
\item[B3] [Solve.] Solve the system $\mathbf x \circ A\mathbf x=\mathbf x$
  in the free variables $x_{n/2+1},\ldots, x_n$ using algorithm {\bf G}.
  If it returns ``none'' proceed to B5.
\item[B4] [Report.] (Algorithm {\bf G} returned the solution set as
  $v,b_1,\ldots, b_d$.)
  Output the vectors $v+\alpha_1b_1+\cdots
  +\alpha_db_d$ for all $2^d$ choices of $\alpha_1,\cdots,\alpha_d\in\{0,1\}$.
\item[B5] [Next.]
  If all choices of $(x_1,\dots,x_{n/2})$ have been inspected,
  terminate.
  Otherwise generate the next choice and return to B3.
\end{description}

The number of choices of $\mathbf x_1$ is $2^{n/2}$, so step B3 takes
total time $2^{n/2}\poly(n)$.
This is dominated by the overall running time of step B4, which is
linear in the number of of contributing subsets $X$.
These are $1.5^n$ in expectation according to lemma \ref{lem: expected
size}.
Hence the runtime is $1.5^n\poly(n)$ in expectation.
This finishes the proof of theorem~\ref{thm: bip}.
\section{General Input Graphs}
\label{sec: genalg}

For general graphs, we do not know how to list the solutions $S$ to
\eqref{eq: quadratic eqs} in step P2 in time proportional to the solution set.
Instead, we will efficiently list a superset $S'$ of the solution set
of size $O(1.619^n)$.
We can then examine every $\mathbf x\in S'$ to see that it solves
\eqref{eq: quadratic eqs}.

The superset $S'$ of candidate solutions to \eqref{eq: quadratic eqs}
is defined as follows.
Let $\mathbf x'$ denote a \emph{prefix} vector of fixed values $\mathbf
x'=(x_1,\dots,x_{n-k})\in\{0,1\}^{n-k}$ and introduce the variables
$x_{n-k+1},\ldots,x_n$.
Set $\mathbf x =(x_1,\ldots, x_n)$.
Consider the equation system
\begin{equation}\label{eq: larger system}
(\mathbf x', \mathbf 0_k)\circ A \mathbf x
= (\mathbf x',\mathbf 0_k)\,.
\end{equation}
For fixed $\mathbf x'$, this is nonhomogeneous set of linear equations in the variables
$x_{n-k+1},\allowbreak \ldots, x_n$.
Every prefix of a solution to \eqref{eq: quadratic eqs} is a solution to
\eqref{eq: larger system}:

\begin{lem}
  Let $\mathbf x=(x_1,\ldots,x_n)\in\{0,1\}^n$ and $\mathbf
  x'=(x_1,\ldots,x_{n-k})$ for some $k\in\{0,1,\ldots, n\}$.
  If $\mathbf x$ satisfies $\mathbf x\circ A \mathbf x=\mathbf x$
then $(\mathbf x',\mathbf 0_k)\circ A \mathbf x=(\mathbf x',\mathbf 0_k)$.
\end{lem}

\begin{proof}
  For $1\leq i\leq n-k$, the $i$th equation of both $ \mathbf x\circ
  A\mathbf x=\mathbf x$ and $(\mathbf x',\mathbf 0_k)\circ A \mathbf
  x=(\mathbf x',\mathbf 0_k)$ is
  \[
  x_i \sum_{j=1}^n a_{ij}x_j = x_i \,.
  \]
  For $i>n-k$, the $i$th equation of $(\mathbf x',\mathbf
  0_k)\circ A \mathbf x=(\mathbf x',\mathbf 0_k)$ simplifies to $0=0$,
  which is trivially satisfied.
\end{proof}

We will show that we can avoid generating all prefixes.

Let $\binom{[n-k]}{k}$ denote the family of vectors $(x_1,\ldots,
x_{n-k})\in \{0,1\}^{n-k}$ for which $x_1+\cdots+x_{n-k} = k$.
Let $\binom{[n-k-1]}{k}_1$ denote the family of vectors $(x_1,\ldots,
x_{n-k})\in \{0,1\}^{n-k}$ for which $x_1+\cdots+x_{n-k-1} = k$
and $x_{n-k}=1$.
The notation is motivated by the cardinalities,
\[
\biggl|\binom{[n-k]}{k}\biggr| = \binom{n-k}{k}\qquad\text{and}\qquad
\biggl|\binom{[n-k-1]}{k}_1\biggr| = \binom{n-k-1}{k}\,.
\]

\begin{lem}
  Let $\mathbf x\in\{0,1\}^n$.
  Then there is exactly one $k\in \{0,\ldots, \lfloor n/2\rfloor\}$ such
  that the $(n-k)$th prefix $\mathbf x'=(x_1,\ldots,x_{n-k})$ belongs to
  $\binom{[n-k]}{k}$  or $\binom{[n-k-1]}{k}_1$.
\end{lem}

\begin{proof}
Let $w_k= x_1+\cdots+x_{n-k}$.
Choose the smallest $k\geq 0$ with the property
\begin{enumerate}
\item[({\it i}\/)] $w_k=k$, or
\item[({\it ii}\/)] $w_k=k+1$, $x_{n-k}=1$.
\end{enumerate}
To see that such $k$ exists consider the pairs of values $(w_k,k)$ for
$k=0,1,\ldots$.
The first pair is $(x_1+\cdots+x_n,0)$.
The left value decreases monotonically by $0$ or $1$ down to $0$.
The right value increases $0,1,2,\ldots$.
The only possibility that ({\it i}) lacks a solutions is that $w_k$
and $k$ pass each other, i.e., there exists $k$ such that
\[
w_{k}=k+1, w_{k+1}=k\,.
\]
But in that case $x_{n-k}=1$ so that ({\it ii}) is satisfied.
Note also that if ({\it ii}) holds for $k$, then ({\it i}) does not
for that $k$, and \emph{vice versa}.
Also note that $k$ is at most $\lfloor n/2\rfloor$,
since $w_k$ is at most $n-k$.

We have that either ({\it i}) is satisfied and $\mathbf x'$ belongs to
$\binom{[n-k]}{k}$, or ({\it ii}) is satisfied and $\mathbf x'$ belongs to
$\binom{[n-k-1]}{k}_1$.

\medskip We claim that there cannot be another solution $k'>k$ to
either ({\it i}) or ({\it ii}).
Consider first the case that $k'$ satisfies ({\it i}) and $k$
satisfies ({\it ii}).
In particular, $k+1=w_k=w_{k-1} +x_{n-k} = w_{k-1} + 1$, so
$k=w_{k-1}$.
Then,
\[
k'=w_{k'}\leq w_{k-1} = k\,,
\]
a contradiction.
The remaining three cases follow from the monotonicity of the left and right hand sides.
\end{proof}

\begin{figure}[bt]
\input{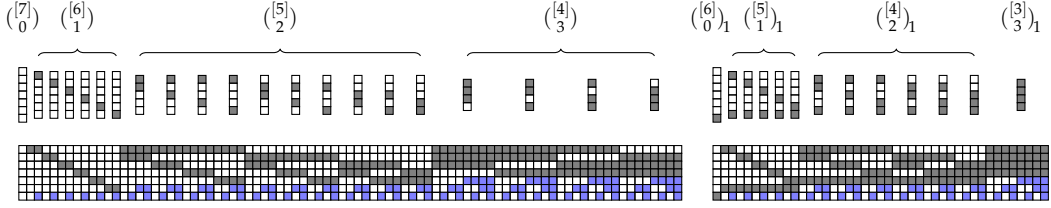}
\caption{\label{fig: bits}The sets $\binom{[n-k]}{k}$ and  $\binom{[n-k-1]}{k}_1$ for
  $n=7$ and $k\in \{0,\ldots, \lfloor n/2\rfloor\}$. The bottom row shows
  all 128 bit patterns on 7 bits.}
\end{figure}

We note in passing that the previous lemma gives a combinatorial proof
of the identity
\[ 
\sum_k
\biggl( \binom{n-k}{k} + \binom{n-k-1}{k}\biggr ) 2^k = 2^n\,,
\]
see figure~\ref{fig: bits}.
For another way to verify this expression, compute
\[ \sum_k \binom{N-k}{k} 2^k = \sfrac{1}{3}(2^{N+1}+(-1)^{N+1})\,, \]
which comes from evaluating a known, closed form for the generating
function $\sum_k\binom{N-k}{k} z^k$ at $z=2$ \cite[exercise
1.2.9.15]{Knuth}.
Then add these values for $N=n$ and $N=n-1$ to arrive at
$\sfrac{1}{3}\bigl(2^{n+1} + (-1)^{n+1} + 2^n + (-1)^n\bigr) =
\frac{1}{3}(2+1)2^n=2^n$.

\medskip
Finally, we need to bound the number of solutions to \eqref{eq: larger
  system}.
\emph{A priori}, this is a system in $k$ variables each ranging over
$\{0,1\}$, so the number of solutions can be $2^k$.
However, we will show that in expectation, the number of solutions is
exactly 1:

\begin{lem}\label{lem: expected no sols}
  Let $\mathbf x'$ denote a vector of fixed values $\mathbf x' = (x_1,
  \ldots, x_{n-k})$ from $\binom{[n-k]}{k}$ or $\binom{[n-k-1]}{k}_1$ and
  introduce the variables $x_{n-k+1},\ldots,x_n\in \{0,1\}$.
  Set $\mathbf x=(x_1\ldots,x_n)$.
  Let $A$ be a matrix whose diagonal entries are chosen uniformly at
  random from $\{0,1\}$.
  Then the expected number of solutions to \[(\mathbf x',\mathbf
  0_k) \circ A \mathbf x=(\mathbf x',\mathbf 0_k) \] is $1$.
\end{lem}

\begin{proof}
  Set $r_i=a_{ii}$.
  For $0\leq i\leq n-k$, the $i$th equation has the form
  \[
  x_i  \sum_{j=1}^n a_{ij}x_j = x_i\,.
  \]
  Since there are exactly $k$ entries $x_i=1$, there are exactly
  $k$ of such equations that do not trivialise to $0=0$.
  We can rewrite each of these equations to
  \[
  \sum_{j=1}^n a_{ij}x_j = 1
  \]
  or, isolating the variables on the left hand side,
  \[
  \sum_{j=n-k+1}^n a_{ij} x_j = 1 + r_i + \sum_{\substack{j=1\\j\neq i}}^{n-k} a_{ij}x_j \,.
  \]
  Now it is clear that the system consists of $k$ equations in the $k$
  variables $x_{n-k+1},\allowbreak\ldots,\allowbreak x_n$.
  The right hand side is uniformly distributed from $\{0,1\}$, so an
  assignment to the variables satisfies an equation with probability
  $\frac{1}{2}$.
  By independence of the choices of $r_i$, all $k$ equations are
  satisfied with probability $2^{-k}$.
  There are $2^k$ assignments, so the expected number of solutions
  is $2^k2^{-k}= 1$.
  \end{proof}

\bigskip\noindent
{\bf Algorithm C} (\emph{Contributing subsets for general graphs.})
Given a graph $G=(V,E)$ with $n=|V|$, lists the vertex subsets
$X\subseteq V$ for which $\degree_x(X)$ is odd for all $x\in X$.
\begin{description}\setlength\itemsep{0pt}
\item[C1] [Add random self loops.] Choose $r_1,\ldots, r_n\in \{0,1\}$
  uniformly at random and add a self-loop at vertex $v_i$ if $r_i=1$.
  Let $A$ denote the adjacency matrix of the resulting graph.
\item[C2] [Initialize.] Set $k=0$, $\mathbf x'=\mathbf 0_n$.
\item[C3] [Solve.]
  Solve the nonhomogeneous linear equation \eqref{eq: larger system}
  in the free variables $x_{n-k+1},\ldots, x_n$ using algorithm~{\bf
    G}.
  If there is no solution, proceed to C5.
\item[C4] [Filter solution.]
  Generate every vector $\mathbf x=v+\alpha_1b_1+\cdots +\alpha_db_d$ for all
  $2^d$ choices of $\alpha_1,\cdots,\alpha_d\in\{0,1\}$.
  If $\mathbf x$ solves the equation $\mathbf x\circ A\mathbf
  x=\mathbf x$, then output $\mathbf x$.
\item[C5] [Next prefix.]
  Update $\mathbf x'$ and possibly $k$ so that $\mathbf x'$ is the
  next element of $\binom{[n-k]}{k}$ or $\binom{[n-k-1]}{k}_1$ for
  $k=0,\ldots,\sfrac{1}{2}n$, and return to C3.
  (When all these elements have been generated, the algorithm terminates.)
\end{description}

\begin{lem}
The running time of algorithm C is within a polynomial factor
of \begin{equation}\label{eq: fibo}
\biggl(\frac{1+\sqrt 5}{2}\biggr) ^n
\end{equation}
in expectation.
\end{lem}

\begin{proof}
  For the running time, first observe that the members of
  $\binom{[n-k]}{k}$ and $\binom{[n-k-1]}{k}_1$ can be listed with
  polynomial (in fact, constant) delay \cite[sec. 7.2.1.3]{Knuth}.

  The number of iterations is given by the sizes of $\binom{[n-k]}{k}$
  and $\binom{[n-k-1]}{k}_1$, summed over all $k=0,\ldots,\lfloor
  n/2\rfloor$.
  These are well-studied quantities, and it is known
  \cite[ex. 1.2.8.16]{Knuth} that
  \[
  \sum_k \binom{n-k}{k} +
  \sum_k \binom{n-k-1}{k} = F_{n+1} + F_n = F_{n+2}\,,\]
  where $F_n$ is the $n$th Fibonacci number, bounded by
  \eqref{eq: fibo}.

  The running time is dominated by the time spent in steps C3 and C4.
  Step C3 takes polynomial time every time it is executed thanks to
  the polynomial running time of algorithm~G.
  Step C4 looks more difficult to bound, because it may exhaust the
  entire solution space to an equation system in $k$ unknowns.
  Naively, this would lead to a total running time of 
  \[\sum_k\binom{n-k}{k}2^k>2^{n-1}\,.\]
  However, by lemma~\ref{lem: expected no sols}, the expected size of
  this solution space is 1.
  Thus, we expect to generate and verify only one vector $\mathbf x$,
  so the time spent in C4 is polynomial in expectation in each
  iteration.
  By linearity of expectation, the total time spent in C4 is within a
  polynomial factor of \eqref{eq: fibo}.
\end{proof}

\subsection{Derandomization}

We use the method of conditional expectations to derandomize
algorithm~C.
Our randomness is over the choices of $r_1,\ldots, r_n$, the self-loop
indicators.
The expected number of candidate solutions $S$ given a specific choice
of values for the first $k$ choices satisfies
\[
\mathbf E [\,S\mid (r_1,\ldots ,r_k)\,]=\sfrac{1}{2}\mathbf E[\,S\mid
(r_1,\ldots, r_k,0)\,]+\sfrac{1}{2}\mathbf E[\,S\mid (r_1,\ldots, r_k, 1)\,]
\]
for all $k\geq 0$.
In particular, there is a specific assignment to $r_1,\ldots,r_n$ such
that $ \mathbf E[\,S\mid (r_1,\ldots,r_n)\,] \leq \mathbf E[S]$.
Thus, by mimicking the behaviour of algorithm~C, but computing the
expected solution sizes of the systems \eqref{eq: larger system}
instead of solving them, we can choose the assignment greedily one
variable at the time, until we get rid of all the randomness.
The final assignment has no more candidate solutions than the original
expected number.

The next lemma shows that we can efficiently compute the expectations
$\mathbf E[\,S\mid (r_1,\ldots, r_k,0)\,]$ and $\mathbf E[\,S\mid (r_1,\ldots, r_k,1)\,]$.

\begin{lem}
  Let $\mathbf x'$ and $\mathbf x$ be an in lemma~\ref{lem: expected
    no sols}. Let $l\in\{0,\ldots,n\}$.
  Let $A$ be a matrix whose diagonal entries $a_{ii}=r_i$ for $l<i\leq
  n$ are chosen uniformly at random from $\{0,1\}$.
  Then the expected number of solutions to \[(\mathbf x',\mathbf
  0_k)\circ A \mathbf x=(\mathbf x',\mathbf 0_k) \] is given by a
  polynomial size linear equation system in the variables
  $x_{n-k+1},\ldots,x_n$ and $r_{l+1},\ldots, r_n$.
\end{lem}

\begin{proof}
  The only contributing equations are those for which $x_i=1$, of the form
  \[ 1\cdot \sum_{j=1}^n a_{ij} x_j = 1 \,.\]
  We can rewrite these equations as
  \[
  \sum_{j=n-k+1}^{n} a_{ij}x_j = 1 + \sum_{j=1}^{n-k} a_{ij} x_j
  \]
  for $i\leq l$, and
  \[
r_i+  \sum_{j=n-k+1}^{n} a_{ij}x_j = 1 + 
  \sum_{\substack{j=1\\j\neq i}}^{n-k} a_{ij} x_j
  \]
  for $i> l$.
  Now it is clear that the system consists of $k$ equations in the $k+(n-l)$
  variables $x_{n-k+1},\ldots, x_n$ and $r_{l+1},\ldots, r_n$.

  If this system has $S$ solutions then the expected number of
  solutions to \eqref{eq: larger system} is $S/2^{n-l+1}$.
\end{proof}

\bigskip\noindent
{\bf Algorithm D} (\emph{Derandomization of C.})
Given a graph $G=(V,E)$ with $n=|V|$, determines values  $r_1,\ldots,
r_n\in\{0,1\}$ for algorithm C. Let $A$ denote the adjacency matrix of $G$.
\begin{description}\setlength\itemsep{0pt}
\item[D1] [Initialize.] Set $l=1$.
\item[D2] [Initialize.] Set $k=0$, $\mathbf x'=\mathbf 0_n$, $N_0=N_1=0$.
\item[D3] [Solve for each choice of $r_l$.]
 For $b=0$, tentatively set $a_{ll} = b$ and solve the nonhomogeneous linear
  equation \eqref{eq: larger system} in the free variables
  $x_{n-k+1},\ldots, x_n$ and $r_{l+1},\ldots, r_n$ using
  algorithm~{\bf G}. Let $S_b$ denote the size of the solution
  space. Repeat step D3 for $b=1$.
\item[D4] [Tally solutions.] Increase $N_b$ by $S_b/2^{n-l+1}$ for $b=0,1$.
\item[D5] [Next prefix.]
  Update $\mathbf x'$ and possibly $k$ so that $\mathbf x'$ is the
  next element of $\binom{[n-k]}{k}$ or $\binom{[n-k-1]}{k}_1$ for
  $k=0,\ldots,\sfrac{1}{2}n$, and return to D3.
  (When all these elements have been generated, proceed to D6.)
\item[D6] If $N_0>N_1$ then fix $a_{ll}=1$, otherwise $a_{ll}=0$.
  Increase $l$.
  If $l\leq n$ return to D2.
  Otherwise output $a_{11},\ldots,a_{nn}$.
\end{description}

\section{Proof of Theorem \ref{thm: comb}}
\label{sec: comb}

Let $\scr C\subseteq 2^E$ denote the set of (directed) cycle covers of
the directed graph $G$.
We use Iverson's bracket notation: for a proposition $P$, we write 
\[
[P] = 
\begin{cases} 
  1, &\text{if $P$}\,;\\
  0, &\text{otherwise}\,.
\end{cases}
\]
$G[Y]$ is the graph induced by $Y\subseteq V$.
We sometimes write $e\in G[Y]$ to refer to a directed edge $e$ in the induced
graph $G[Y]$.
We will operate on sums of partitions of vertex sets, and manipulate
these sums. For this purpose, we introduce the notation 
\[\sum_{Y_1,\ldots, Y_K}^V \]
for the sum over all ordered $K$-partitions $Y_1,\ldots, Y_K$ of $V$. That
is, $Y_1\cup\cdots\cup Y_K=V$ and $Y_i\cap Y_j=\emptyset$ ($1\leq
i<j\leq k$). 
Parts may be empty.

For a subset $T\subseteq E$ of directed edges define
\begin{equation*}
  h(T) = \frac{1}{K} \sum_{Y_1,\ldots,Y_K}^V \prod_{e\in T} 
  [ e\in G[Y_1]\cup\cdots\cup G[Y_K]]\,.
\end{equation*} 

We will show that the residue modulo $K$ of the function $h$ serves as
an indicator variable for Hamiltonicity on $\scr C$.

\begin{lem}
  For $C\in \scr C$,
  \begin{equation}
    [C\in \scr H] =    h(C) \pmod K\,.
  \end{equation}
\end{lem}

\begin{proof}
  First, we show that for $C\in \scr H$ we have $h(C)=1 \,\pmod K\,$.
  Consider a Hamiltonian cycle $C$ and a partition $Y_1,\ldots, Y_K$
  of $V$.
  First, every partition that sets $Y_k=V$ for some $k\in \{1,\ldots,
  K\}$ (and all other parts empty) has $e\in G[Y_k]$ for all $e\in
  C$, so its contribution is 1.
  On the other hand, consider a partition where some $Y_k$ ($1\leq
  k\leq K$) is neither $\emptyset$ nor $V$.
  Since $C$ is Hamiltonian, there is a directed edge $uv\in C$ with $u\in Y_k$
  and $v\notin Y_k$.
  This directed edge belongs to none of the $G[Y_k]$, so the product vanishes
  and the partition does not contribute to the sum.
  Thus, the total contribution of $C$ to the sum in $h(C)$ is $K$, and $h(C)=1 \pmod K$.

  \bigskip Second, we show for each $C\notin \scr H$ that $h(C)$ is
  $0\pmod K$.
  In particular, we show
  \[\sum_{Y_1,\ldots,Y_K}^V \prod_{e\in C} 
  [e\in G[Y_1]\cup\cdots\cup G[Y_K]] = 0 \pmod {K^2}\,.\] 
  Partition the non-Hamiltonian cycle cover $C$ into cycles 
  \[C=C_1 \cup \cdots \cup C_r\,,\] 
  such that each $C_i$ is a simple cycle. 
  Note that $r>1$ because $C\notin\scr H$.
  
  Construct the corresponding partition of vertices
  \[V=V_1\cup \ldots\cup V_r\,,\]
  such that $V_i$ are the vertices visited by the cycle $C_i$.
  
  Let $Y_1,\ldots, Y_K$ be a vertex partition such that
  \begin{equation}\label{eq: Y contributes} \prod_{e\in C}
    [ e \in G[Y_1]\cup\cdots\cup G[Y_K]] = 1\,.
  \end{equation}
  We
  first need to observe that the cycle-induced vertex partition
  $V_1,\ldots, V_r$ refines $Y_1,\ldots, Y_K$.
  Indeed, assume that $V_i\cap Y_j$ is neither empty nor $V_i$.
  Then there is a directed edge $e=uv\in C_i$ with $u\in V_i\cap Y_j$ but
  $v\notin V_i\cap Y_j$.
  In particular, the directed edge $e$ belongs to neither $G[Y_j]$ nor to any other of the
  $G[Y_k]$, contradicting \eqref{eq: Y contributes}.
  Hence the partition $Y_1,\ldots, Y_K$ consists of the $r$ parts of
  $V_1,\ldots, V_r$.
  There are $K^r$ ways to pick these parts, each amounting to one $1$
  in the summation.
  Since $r>1$ this shows that the total contribution of $C\notin \scr
  H$ is a multiple of $K^2$ and after division of $K$ still vanishes
  modulo $K$.
\end{proof}

In particular, we can count the number of Hamiltonian cycles modulo $K$ as
\begin{equation}\label{eq: H}
|\scr H| = \sum_{C\in\scr C} h(C) \pmod K\,.
\end{equation}
We next rewrite the right hand side by an application of
inclusion--exclusion.

For a vertex set $Z\subseteq V$ let $\scr F(Z)$ be the family of edge
subsets in which every vertex has outdegree $1$ and all terminals are in
$Z$.
In other words, $\scr F(Z)$ is the family of total functions
$f\colon V\rightarrow Z$ where $vf(v)$ is a directed edge in the graph for all $v\in V$.

\begin{lem}\label{lem: ie}
\[\sum_{C\in\scr C} h(C) = \sum_{Z\subseteq V} (-1)^{|V\setminus Z|}\sum_{F\in \scr F(Z)}
h(F)\,\pmod K.\]
\end{lem}

\begin{proof}
  Consider a cycle cover $C\in \scr C$.
  It belongs to $\scr F(Z)$ exactly when $Z=V$,
  so its total contribution to the right hand side is $h(C)$.

  Consider now an edge subset $F\in \scr F(V) \setminus \scr C$ that is not a cycle
  cover.
  Let $X$ denote the set of vertices appearing as terminals for the
  directed edges in $F$.
  Then $F$ belongs to $\scr F(Z)$ for every $Z$ with $X\subseteq
  Z\subseteq V$.
  By the principle of inclusion--exclusion, the number of $Z$ with
  $X\subseteq Z\subseteq V$ is even (using that $X\subsetneq V$ because
  $F$ is not a cover).
  In particular there are as many odd $Z$ as even ones.
  Thus, the contributions of all $F\in \scr F(V) \setminus \scr C$ cancel.
\end{proof}

Recall that $\degree_v(X)$ denotes the number of directed edges from $v$ to a vertex
in $X$.
For pairwise disjoint subsets $Y_1,\ldots, Y_K$ of $V$ whose union includes all of $Z$, let $\degree_v(Z;
Y_1,\ldots,Y_K)$ denote the number of directed edges from $v$ to a vertex in $Z$
that stay in the same part as $v$, formally
\[
\degree_v(Z; Y_1,\ldots,Y_K) = N(Z\cap Y_k)\,,\quad \text{where $v\in Y_k$}\,.
\]

\begin{lem}\label{lem: rewrite}
  For $Z\subseteq V$, 
\[ \sum_{F\in \scr F(Z)} h(F) = 
\frac{1}{K} \prod_{v\in V\setminus Z} \degree_v(Z)
\sum_{Y_1,\ldots,Y_K}^Z
\prod_{v\in Z} \degree_v(Z;Y_1,\ldots, Y_K)\pmod{K} \,.\]
\end{lem}

\begin{proof}
Expanding $h(F)$ and rearranging, we have
  \begin{align*}
    \sum_{F\in \scr F(Z)} h(F) &=     
    \frac{1}{K} \sum_{F\in \scr F(Z)}
    \sum_{Y_1,\ldots, Y_K}^V
    \prod_{e\in F}
    [
    e\in G[Y_1]\cup\cdots\cup G[Y_K]
    ]\\
&=   \frac{1}{K} \sum_{Y_1,\ldots,Y_K}^V
    \sum_{F\in \scr F(Z)}
    \prod_{e\in F}
    [
    e\in G[Y_1]\cup\cdots\cup G[Y_K]
    ]\\
&=  \frac{1}{K}
    \sum_{Y_1,\ldots,Y_K}^V
    \prod_{v\in V} \degree_v(Z; Y_1,\ldots, Y_K)\,,
 \end{align*}
 where the last step is based on counting in two different ways the
 number of ways that every vertex in $V$ can choose another vertex in
 $Z$ in the same part of $Y_1,\ldots, Y_K$.

\bigskip
The next step is to establish
\begin{equation}\label{eq: induction}
\sum_{Y_1,\ldots,Y_K}^V
\prod_{v\in V} \degree_v(Z;Y_1,\ldots, Y_K) =
\biggl(\prod_{v\in V\setminus Z} \degree_v(Z) \biggr)
\sum_{Y_1,\ldots,Y_K}^Z
\prod_{v\in Z} \degree_v(Z;Y_1,\ldots, Y_K)
\end{equation}
by induction in $|V\setminus Z|$. 

For $V=Z$ the claim is vacuous.
Otherwise, select a vertex $w\in V\setminus Z$ and split the sum on
the left hand side into $K$ sums according to which part includes the
vertex $w$.
\[
\sum_{Y_1,\ldots,Y_K}^V
\prod_{v\in V} \degree_v(Z;Y_1,\ldots,Y_K)= S_1+\cdots + S_K,\]
where
\[
S_k = \sum_{Y_1,\ldots,Y_K}^{V\setminus\{w\}}
\prod_{v\in V}
\degree_v(Z;Y_1,\ldots,Y_k\cup\{w\},\ldots,Y_K)\,.
\]
In the product, the factor contributed by the term corresponding to
$v=w$ is 
\[\degree_w(Z;Y_1,\ldots,Y_k\cup\{w\},\ldots, Y_K) = \degree_w(Z\cap (Y_k\cup
\{w\})) = \degree_w(Z\cap Y_k)\,,\]
because $w$ does not belong to $Z$.
Thus,
\begin{align*}
S_k &
=\sum_{Y_1,\ldots,Y_K}^{V\setminus\{w\}}
\degree_w(Z\cap Y_k)
\prod_{v\in V\setminus\{w\}}
\degree_v(Z;Y_1,\ldots,Y_K)\,.\\
\end{align*}

Summing all the $S_k$ we arrive at
\begin{align*}
S_1+ \cdots+S_K &= \sum_{Y_1,\ldots,Y_K}^{V\setminus \{w\}}
\biggl\{
\biggl (\sum_{k=1}^K \degree_w(Z\cap Y_k)\biggr)
 \prod_{v\in V\setminus\{w\}}
\degree_v(Z;Y_1,\ldots,Y_K)\biggr\}\\
 &= \sum_{Y_1,\ldots,Y_K}^{V\setminus \{w\}}
 \degree_w(Z) \prod_{v\in V\setminus\{w\}}
\degree_v(Z;Y_1,\ldots,Y_K)\\
 &= \degree_w(Z)  \sum_{Y_1,\ldots,Y_K}^{V\setminus \{w\}}
\prod_{v\in V\setminus\{w\}}
\degree_v(Z;Y_1,\ldots,Y_K)\\
&= \biggl(\prod_{w\in V\setminus Z} \degree_w(Z)\biggr)
  \sum_{Y_1,\ldots,Y_K}^Z
\prod_{v\in Z}
\degree_v(Z;Y_1,\cdots,Y_K) 
\,,
\end{align*}
where the last identity follows by induction over $w$, establishing \eqref{eq: induction}.
\end{proof}

We are ready to establish theorem~\ref{thm: comb}, namely,
\begin{equation*}
  |\scr H |  = \frac{1}{K} \sum_{Z,Y_1,\ldots,Y_K}^V (-1)^{|V\setminus  Z|}
  \biggl(\prod_{z\in Z} \degree_z(V\setminus Z)\biggr)
  \prod_{k=1}^K \prod_{y\in Y_k} \degree_y(Y_k)\pmod K\,.
\end{equation*}

\begin{proof}[Proof of theorem~\ref{thm: comb}]
Write $\overline{Z} = V\setminus Z$.
Combining \eqref{eq: H} and lemmas
\ref{lem: ie} and \ref{lem: rewrite} we get:

\begin{align*}
 |\scr H| & = \frac{1}{K} \sum_{Z\subseteq V} \sum_{Y_1,\ldots,Y_K}^Z (-1)^{|\overline
   Z|} \prod_{v\in \overline Z} \degree_v(Z) 
\prod_{v\in Z} \degree_v(Z;Y_1,\ldots,Y_K)  \pmod K
\\
 & = \frac{1}{K}  \sum_{Z,Y_1,\ldots,Y_K}^V
(-1)^{|Z|} \prod_{v\in Z} \degree_v(\overline Z) 
\prod_{v\in \overline Z} \degree_v(\overline Z;Y_1,\ldots,Y_K)  \pmod K
\,,
\end{align*}
where we just changed the original summand  $Z$ to its complement
$\overline Z$.

Finally, every $y\in \overline Z$ belongs to exactly one $Y_k$, and
for this value of $k$ we have $\degree_y(\overline Z; Y_1,\ldots,Y_K) =
\degree_y(Y_k)$.
Thus,
\begin{align*}
\prod_{y\in \overline Z} \degree_y(\overline Z; Y_1,\ldots, Y_K)  = 
\prod_{k=1}^K \prod_{y\in Y_k} \degree_y(Y_k)\,.
\end{align*}
\end{proof}

\end{document}